\documentclass[11pt]{article}

\usepackage{orcidlink}
\usepackage{amsmath, amssymb, amsthm}
\usepackage{bbm}
\usepackage{graphicx}
\usepackage{fullpage}
\usepackage[round]{natbib}
\usepackage{hyperref}
\usepackage{xcolor}
\usepackage{booktabs}
\usepackage{algpseudocode}
\usepackage{algorithm}
\algrenewcommand\algorithmicrequire{\textbf{Input:}}
\algrenewcommand\algorithmicensure{\textbf{Output:}}
\usepackage{bm}

\title{
Langevin-Gradient Rerandomization
}
\author{
    Ant\^onio Carlos Herling Ribeiro Junior\thanks{Department of Statistics and Data Science, Carnegie Mellon University, Pittsburgh, Pennsylvania, United States}
}

\newtheorem{theorem}{Theorem}[section]

\newtheorem{assumption}[theorem]{Assumption}

\newcommand{\Z}{\mathcal{Z}}

\newcommand{\1}{\mathbbm{1}}

\newcommand{\Prob}{\mathbb{P}}
\newcommand{\var}{\text{Var}}
\newcommand{\varcr}{\text{Var}_{\text{CR}}}
\newcommand{\varlgr}{\text{Var}_{\text{LGR}}}
\newcommand{\cov}{\text{cov}}
\newcommand{\E}{\mathbb{E}}

\newcommand{\R}{\mathbb{R}}

\newcommand{\tauest}[1]{\widehat{\tau}_{\text{#1}}}

\newcommand{\mean}[1]{\overline{#1}}

\begin{document}

\maketitle

\begin{abstract}
    Rerandomization is an experimental design technique that repeatedly randomizes treatment assignments until covariates are balanced between treatment groups. Rerandomization in the design stage of an experiment can lead to many asymptotic benefits in the analysis stage, such as increased precision, increased statistical power for hypothesis testing, reduced sensitivity to model specification, and mitigation of p-hacking. However, the standard implementation of rerandomization via rejection sampling faces a severe computational bottleneck in high-dimensional settings, where the probability of finding an acceptable randomization vanishes. Although alternatives based on Metropolis-Hastings and constrained optimization techniques have been proposed, these alternatives rely on discrete procedures that lack information from the gradient of the covariate balance metric, limiting their efficiency in high-dimensional spaces. We propose Langevin-Gradient Rerandomization (LGR), a new sampling method that mitigates this dimensionality challenge by navigating a continuous relaxation of the treatment assignment space using Stochastic Gradient Langevin Dynamics. We discuss the trade-offs of this approach, specifically that LGR samples from a non-uniform distribution over the set of balanced randomizations. We demonstrate how to retain valid inference under this design using randomization tests and empirically show that LGR generates acceptable randomizations orders of magnitude faster than current rerandomization methods in high dimensions.

\end{abstract}

\section{Introduction}
Randomized controlled trials are widely regarded as the ``gold standard" for estimating causal effects in a
wide range of disciplines \citep{mosteller2002evidence, carlin2007introduction, druckman2012cambridge, athey2017econometrics, list2024experimental}. While complete randomization ensures that treatment and control groups are balanced on both observed and unobserved covariates on average, it does not guarantees balance in any single realization \citep{rubin2008comment, rosenberger2008handling}. In finite samples, chance imbalances in baseline covariates can significantly inflate the variance of treatment effect estimators and reduce the power of statistical tests \citep{morgan2012}. Moreover, with as the number of covariates increase, the higher the probability of imbalance between treatment groups \citep{krieger2019nearly, morgan2012}. To address such imbalances, other design of experiments can be considered, such as stratified \citep{miratrix2013adjusting, tabord2022stratification}, blocked \citep{pashley2021insights, pashley2022block}, matched-pair \citep{imai2008variance, bai2022optimality}, and adaptive designs \citep{rosenberger2002randomization, hu2006theory}. Another alternative is through rerandomization—a design strategy that repeatedly discards assignments failing to meet a pre-specified covariate balance criterion—has emerged as a powerful tool to improve estimation and inference efficiency \citep{morgan2012, li2018, li2020rerandomization}.

The concept of rerandomization has been commented on since the 1950s \citep{grundy1950restricted, jones1958inadmissible, savage1962foundations} and has been used by many since then \citep{bruhn2009} although it has not been formalized until the work of \cite{morgan2012}. After \cite{morgan2012}, rerandomization has been extended to many experimental designs, including 2K factorial, stratified, sequential, survey, and split plot \citep{branson2016improving, li2020factorial, zhou2018, yang2021rejective, shi2024rerandomization}. Using rerandomization in the design stage of an experiment has been shown to lead to many asymptotic benefits in the analysis stage, such as increased precision \citep{morgan2012, li2018, li2020rerandomization, wang2024, wang2025asymptotic}, increased statistical power of the hypothesis test \citep{branson2024power},  reduced sensitivity to model specification \citep{zhao2024b}, and thus mitigation of p-hacking \citep{lu2025rerandomization}.

However, rerandomization is typically implemented by acceptance-rejection sampling \citep{morgan2012, li2018, ribeiro2025does}. This implementation faces a severe ``curse of dimensionality." As the number of covariates increases, the probability of a random assignment satisfying the balance criterion vanishes exponentially, making the search for a valid allocation computationally prohibitive for even moderately high-dimensional settings.

Recent work has sought to mitigate this bottleneck using more sophisticated sampling techniques. Pair-Switching Rerandomization (PSRR) \citep{zhu2023} employs a Markov Chain Monte Carlo approach, iteratively swapping treatment assignment between pairs of units until the randomization is deemed balanced. While effective in low-dimensional settings, PSRR essentially operates as a local random walk with a fixed step size. In high-dimensional spaces, where the set of balanced randomizations represents a small region of the discrete hypercube, such local searches often fail to find acceptable assignments in a reasonable timeframe. On the other hand, Balanced Randomization via Integer Programming (BRAIN), a constrained optimization approach proposed by \cite{lu2025fast}, can be much faster in high dimensions, but remains restricted to discrete moves in the treatment assignment space, which prevents it from directly exploiting gradient information of the covariate imbalance metric.

In this paper, we propose Langevin Gradient Rerandomization (LGR). Our approach shifts the problem from a discrete to a continuous sampling task. By relaxing the binary treatment assignment into a continuous latent space, we utilize Stochastic Gradient Langevin Dynamics (SGLD) to follow the gradient of the covariate imbalance measure toward the set of balanced randomizations. Unlike the ``blind" search of rejection sampling, the ``random walk" of PSRR, or the discrete optimization of BRAIN, LGR uses the gradient of the covariate imbalance metric to guide the sampling process.

We make two primary contributions. First, we prove that LGR leads to an unbiased and more precise difference-in-means treatment effect estimator despite sampling non-uniformly from the set of balanced randomizations. Since we sample non-uniformly on the balanced set, we use Fisher randomization tests to conduct finite-sample valid inference. These properties align with the finite-sample inference guarantees established for PSRR and BRAIN. Second, we empirically show that LGR samples balanced randomizations orders of magnitude faster than existing methods as the dimension size increases.

The rest of this paper is divided as follows. In Section~\ref{sec:setup} we cover the basics of rerandomization. Next, in Section~\ref{sec:method} we define LGR, provide its statistical properties, and explain how to conduct finite-sample inference. We present our empirical results in Section~\ref{sec:simulations}. We conclude in Section~\ref{sec:conclusion}.

\section{Setup and Notation}\label{sec:setup}
Consider a randomized experiment with $n$ units, where $n_1$ units are assigned to treatment and $n_0 = n - n_1$ to control, with $r_z = n_z/n, z = 0,1$ the proportion of units under arm $z$. We observe a matrix of covariates $X = (X_1, \dots, X_d)$, where $X_j \in \R^n$. The treatment assignment vector is $\Z = (Z_1, \dots, Z_n)^\prime$, with $Z_i=1$ for treatment and $Z_i=0$ for control. Under the potential outcomes framework \cite{rubin1974}, let $Y_i(1)$ and $Y_i(0)$ be the potential outcomes for unit $i$ under treatment and control, and their respective vectors $Y(1) = (Y_1(1), \dots, Y_n(1))^\prime$ and $Y(0) = (Y_1(0), \dots, Y_n(0))^\prime$. The observed outcome for unit $i$ is $Y_i = Y_i(1)Z_i + Y_i(0)(1-Z_i)$. The sole source of randomness is treatment assignment $Z_i$, while the covariates and potential outcomes are fixed. Ultimately, the goal is to estimate the Average Treatment Effect (ATE), defined by $\tau = \frac{1}{n}\sum_{i=1}^{n}\tau_i = \frac{1}{n}\sum_{i=1}^{n}Y_i(1) - Y_i(0)$, where $\tau_i$ are the individual treatment effects. Following most of the rerandomization literature, we use the difference-in-means estimator to estimate the ATE
\begin{equation*}
    \tauest{} = \frac{1}{n_1}\sum_{i=1}^{n}Y_iZ_i -\frac{1}{n_0}\sum_{i=1}^{n}Y_i(1-Z_i).
\end{equation*}

Rerandomization ensures balance between the treatment and control groups with respect to the observed covariates. Covariate balance can be measured in various ways, but following most of the rerandomization literature, we focus on the Mahalanobis distance
\begin{equation*}
    M = \tauest{X}' \Sigma^{-1} \tauest{X} 
\end{equation*}
where 
\begin{equation*}
    \tauest{X} = \left(\mean{X}_1 - \mean{X}_0\right) = \frac{1}{n_1}X'Z - \frac{1}{n_0}X'(1_n-Z)
\end{equation*}
with $1_n \in \R^n$ is a vector whose coefficients are all equal to one, and $\Sigma = \cov\left(\tauest{X}\right)$. Under complete randomization, we can write $\Sigma = \frac{n}{n_1n_0}S^2_X$ where $S^2_X = \frac{1}{n-1}\sum_{i=1}^{n}(X_i - \mean{X})(X_i - \mean{X})'$ is the sample covariance matrix of the covariates.

In a rerandomized design, a randomization $Z$ is deemed balanced only if $M \leq a$, where $a>0$ is a pre-specified threshold. Hence, we define the set of balanced randomizations as $\Z_a = \{Z:M \leq a\}$. Asymptotically, the Mahalanobis distance follows a $ \chi^2_d$ distribution \citep{morgan2012}. Therefore, it is common to set $a$ based on a acceptance probability $p_a = \Prob(\chi^2_d \leq a)$, where $p_a$ is commonly chosen to be $0.01$ or $0.001$.

\section{Langevin-Gradient Rerandomization}\label{sec:method}
\subsection{The LGR Algorithm}

To address the computational bottleneck of finding a balanced assignment $Z\in \Z_a$ in high-dimensional settings, we propose Langevin-Gradient Rerandomization (LGR). Unlike rejection sampling, which blindly draws randomizations from all possible treatment assignments, LGR utilizes the gradient of the Mahalanobis distance with respect to a continuous relaxation of the treatment assignment to actively guide the sampling towards the set of balanced randomizations $\Z_a$.


We introduce a vector of latent scores $\theta \in \mathbb{R}^n$, which are mapped to ``soft" assignments $\tilde{z} \in (0, 1)^n$ via a temperature-scaled logistic function:
\begin{equation}\label{eq:soft_treatment}
    \tilde{z}_i(\theta_i) = \sigma_\delta(\theta_i) = \frac{1}{1 + \exp(-\theta_i / \delta)}
\end{equation}

where $\delta > 0$ controls the smoothness of the relaxation. This relaxation allows us to define a differentiable ``soft" Mahalanobis distance, which is the Mahalanobis distance calculated with respect to the soft assignments $\tilde{z}$. The gradient of $M$ with respect to the latent scores is calculated via the chain rule:
\begin{equation}\label{eq:gradient}
    \frac{\partial M}{\partial \theta_i} = \frac{\partial M}{\partial \tilde{z}_i} \cdot \frac{\partial \tilde{z}_i}{\partial \theta_i},
\end{equation}
where the derivative of the sigmoid is $\frac{\partial \tilde{z}_i}{\partial \theta_i} = \frac{1}{\delta} \tilde{z}_i (1 - \tilde{z}_i)$.

The gradient of $M$ with respect to the soft weights $\tilde{z}_i$, after applying the quotient rule to the difference-in-means terms, is:
\begin{equation}\label{eq:grad-M}
    \frac{\partial M}{\partial \tilde{z}_i} = 2 \left[ \widehat{\Sigma}^{-1} \Delta(\tilde{z}) \right]^\top \left[ \frac{1}{n_{\tilde{1}}}(X_i - \mean{X}_{\tilde{1}}) + \frac{1}{n_{\tilde{0}}}(X_i - \mean{X}_{\tilde{0}}) \right],
\end{equation}
where $n_{\tilde{1}} = \sum \tilde{z}_i$ and $n_{\tilde{0}} = \sum (1 - \tilde{z}_i)$ are the soft treatment group sizes, and $\mean{X}_{\tilde{1}}$ and $\mean{X}_{\tilde{0}}$ are the covariate means with respect to the soft assignments. 


The core of LGR, detailed in Algorithm \ref{alg:lgr}, is the iterative evolution of the latent scores $\theta$ using Stochastic Gradient Langevin Dynamics (SGLD). The algorithm is initialized with $\theta^{(0)} \sim N(0, I_n)$. At each iteration $t$, we update the scores according to:$$\theta^{(t)} \leftarrow \theta^{(t-1)} - \eta \nabla_\theta M\left(\theta^{(t-1)}\right) + \sqrt{2\eta\delta}\xi_t$$where $\eta > 0$ is the learning rate and $\xi_t \sim N(0, I_n)$ is standard Gaussian noise. This update rule consists of two competing forces: the gradient term $-\eta \nabla_\theta M$ drives the latent scores toward values that minimize covariate imbalance, while the noise term $\sqrt{2\eta\delta}\xi_t$ injects stochasticity. This stochastic component is crucial; it prevents the algorithm from collapsing into a deterministic optimization and prohibiting randomization-based inference.

\begin{algorithm}[h]
\caption{Langevin-Gradient Rerandomization}
\label{alg:lgr}
\begin{algorithmic}[1]
\Require 
    Covariates $X$, 
    Number of treated units $n_1$, 
    Balance threshold $a$,  
    Temperature $\delta$ (default value = 0.5),
    Learning rate $\eta$ (default value = 1)
\Ensure A balanced assignment $Z$ where $\sum Z_i = n_1$ and $M \le a$

\State \textbf{Initialize:} $\theta^{(0)} \sim N(0, I_n)$
\State Pre-compute $\widehat{\Sigma}^{-1}$

\While{$M > a$}
    \State \textbf{1. Discrete Projection Check}
    \State Let $\mathcal{S}$ be the indices of the $n_1$ largest elements of $\theta^{(t-1)}$
    \State Construct binary vector $Z$: $Z_i = 1$ if $i \in \mathcal{S}$, else $0$
    \State Compute $M$
    \If{$M \le a$}
        \State \Return $Z$ \Comment{Found balanced randomization}
    \EndIf

    \State \textbf{2. Soft Relaxation}
    \State $\tilde{z} \leftarrow \sigma_\delta(\theta^{(t-1)})$
    \State $n_{\tilde{1}} \leftarrow \sum \tilde{z}_i, \quad n_{\tilde{0}} \leftarrow n - n_{\tilde{1}}$

    \State \textbf{3. Gradient Computation}
    \State $\Delta \leftarrow \left(\frac{1}{n_{\tilde{1}}} X'\tilde{z}\right) - \left(\frac{1}{n_{\tilde{0}}} X'(\1_n-\tilde{z})\right)$
    \State $g_{M} \leftarrow 2 \cdot \widehat{\Sigma}^{-1} \Delta$
    \State Scaling vector $\gamma\in \mathbb{R}^n$ where $\gamma_i = \frac{1}{\delta} \tilde{z}_i (1 - \tilde{z}_i)$
    \State $\nabla_{\theta} \leftarrow (\text{Jacobian of } M \text{ w.r.t } \tilde{z}) \cdot \gamma$ \Comment{Using equation (\ref{eq:gradient})}

    \State \textbf{4. Update Step (SGLD)}
    \State $\xi_t \sim N(0, I_n)$
    \State $\theta^{(t)} \leftarrow \theta^{(t-1)} - \eta\nabla_{\theta} + \sqrt{2\eta\delta}\xi_t$
\EndWhile

\end{algorithmic}
\end{algorithm}

While the SGLD updates occur in the continuous latent space, our goal is a binary assignment. We construct a candidate binary assignment $Z$ by assigning the treatment to the units corresponding to the $n_1$ largest elements of $\theta^{(t)}$. If this candidate $Z$ satisfies the balance criterion $M \le a$, the algorithm terminates and returns $Z$.

The efficiency of LGR relies on the choice of the temperature $\delta$ and the learning rate $\eta$. We set the default temperature to $\delta = 0.5$. Since the latent scores are initialized from a standard normal distribution, $\theta^{(0)} \sim N(0, I_n)$, the input to the sigmoid function $\theta/\delta$ effectively follows a distribution with variance $1/\delta^2 = 4$. This scaling spreads the initial scores across the active domain of the logistic function, preventing the soft assignments from saturating at 0 or 1 too early, which would cause the gradients to vanish. We set the default learning rate to $\eta = 1$. This value provides a balanced step size that allows for convergence toward the balanced region while maintaining sufficient variance in the Langevin noise to ensure valid coverage of the assignment space.

\subsection{Statistical properties of LGR}

\begin{assumption}\label{ass:same_size}
$n_1 = n_0 = n/2$.
\end{assumption}

\begin{assumption}\label{ass:linearity}
$Y_i(Z_i) = \beta_0 + \beta^\prime X_i + \tau Z_i + \epsilon_i$, where $\beta_0 + \beta^\prime X_i$ is the linear projection of $Y_i(0)$ onto $(1, X)$ and $\epsilon_i$ is the deviation from the linear projection.
\end{assumption}

\begin{assumption}\label{ass:normality}
$\tauest{}$ and $\mean{X}_1 - \mean{X}_0$ are normally distributed.
\end{assumption}

\begin{theorem}[Unbiasedness]\label{thm:unbiased}
    Under Assumption~\ref{ass:same_size} and LGR, we have that the difference-in-means estimator is unbiased: $\mathbb{E}[\hat{\tau}] = \tau$.
\end{theorem}

\begin{theorem}[Variance Reduction]\label{thm:variance}
    Under Assumptions~\ref{ass:same_size}-\ref{ass:normality} and LGR, we have that
    \begin{equation*}
        \frac{\varcr(\tauest{}) - \varlgr(\tauest{})}{\varcr(\tauest{})} \geq \left(1 - \frac{a}{d}\right)R^2
    \end{equation*}
    where
    \begin{equation*}
        R^2 = \frac{n}{n_1n_0}\frac{\beta^\prime S^2_X\beta}{\varcr(\tauest{})},
    \end{equation*}
    $\varcr(\tauest{})$ the variance of $\tauest{}$ under complete randomization, and $\varlgr(\tauest{})$ the variance of $\tauest{}$ under LGR.
\end{theorem}

Taken together, Theorems~\ref{thm:unbiased}-\ref{thm:variance} show that, from a design-based perspective, LGR enjoys the same key finite-sample properties as existing rerandomization schemes such as PSRR and BRAIN. Thus, the main distinction between LGR and these approaches lies in the mechanism used to sample the space of balanced assignments, where LGR replaces discrete local moves with a gradient-guided Langevin sampler.

\subsection{Inference under LGR}

A key challenge in rerandomization designs that do not sample uniformly from the balanced set $\Z_a$—such as PSRR, BRAIN, and our proposed LGR—is that standard asymptotic results for rerandomization (e.g., \cite{li2018}) may not hold. Specifically, because LGR utilizes gradient information to steer the sampling, the resulting distribution of balanced assignments $P(Z|Z \in \Z_a)$ is not necessarily uniform. Consequently, standard theoretical results on asymptotic distributions derived in \cite{morgan2012, li2018} are not directly applicable.

To ensure valid inference without relying on asymptotic approximations that may be violated by the non-uniform sampling, we employ Fisher Randomization Tests (FRT). The FRT provides exact finite-sample inference by simulating the distribution of the test statistic under the sharp null hypothesis of no treatment effect, conditional on the specific randomization mechanism employed.

Given the vectors of potential outcomes $Y(1)$ and $Y(0)$, we wish to test the sharp null hypothesis $H_0: Y_i(1) = Y_i(0)$ for all $i=1, \dots, n$. Under $H_0$, the observed outcomes $Y$ are fixed regardless of the treatment assignment. We define a test statistic: $T(Z, Y)$. After sampling $B \gg 0 $ balanced randomizations $Z^{(1)}, \dots, Z^{(B)} \in \Z_a$, we define the p-value as
\begin{equation*}
    p = \frac{1}{B+1} \left[1 + \sum_{b=1}^{B} \1 \left( |T(Z^{(b)}, Y)| > |T(Z, Y)| \right)\right]
\end{equation*}

where $\1$ is the indicator function.



To construct confidence intervals for the average treatment effect, we invert this test. While the standard FRT tests the sharp null hypothesis of no effect ($H_0: \tau = 0$), we can extend this to test any constant additive treatment effect hypothesis $H_0(\tau_0): Y_i(1) - Y_i(0) = \tau_0$ for all $i$, $\tau_0 \in \R$.

Under the hypothesis $H_0(\tau_0)$, the missing potential outcomes are imputed as:
\begin{equation*}
    Y_i(0) = \begin{cases} Y_i & \text{if } Z_i = 0 \\ Y_i - \tau_0 & \text{if } Z_i = 1 \end{cases} \quad \text{and} \quad Y_i(1) = \begin{cases} Y_i + \tau_0 & \text{if } Z_i = 0 \\ Y_i & \text{if } Z_i = 1 \end{cases}
\end{equation*}

This allows us to construct the full vector of potential outcomes under the null and generate the reference distribution of the test statistic $T(Z, Y(\tau_0))$ by repeatedly applying the LGR algorithm. The $(1-\alpha)$ confidence interval is defined as the set of all values $\tau_0$ such that the null hypothesis $H_0(\tau_0)$ is not rejected at level $\alpha$:
\begin{equation*}
    \text{CI}_{1-\alpha} = \{ \tau_0 : p(\tau_0) \ge \alpha \}
\end{equation*}
where $p(\tau_0)$ is the p-value obtained from the FRT under the hypothesized effect $\tau_0$. In practice, this interval is approximated via a grid search over plausible values of $\tau$. While this approach is computationally intensive for standard rejection sampling, the efficiency of LGR renders this inversion feasible even in high-dimensional settings.

\section{Simulations}\label{sec:simulations}

We evaluate the performance of LGR against standard Complete Randomization (CR), Acceptance Rejection Sampling Rerandomization (ARR), Pair-Switching Rerandomization (PSRR), and the BRAIN algorithm. We compare them in terms of (i) computational time (in seconds) to find a balanced randomization -- where for CR is simply the time to draw a randomization and serve as a benchmark,  (ii) bias and standard deviation of the treatment effect estimator, (iii) coverage and power of confidence intervals. We implement the PSRR algorithm as described in Algorithm 1 of \cite{zhu2023}), BRAIN according to Algorithm 1 of \cite{lu2025rerandomization}. In both of them, we use the default suggested values for the tuning parameters. We also use the default parameters' values to report the results of LGR. We run the simulations in a MacBook Air (Apple M2 Chip, 24 GB Memory, 8 Cores)\footnote{All code for the simulations is available in this \href{https://github.com/antonioherling/lgr-rerandomization}{github repository}.}.

We simulate a setting with $d$ covariates generated from a multivariate normal distribution $X \sim N(0, I_d)$. 
The outcome follows a linear model $Y_i(Z_i) = \sum_{j=1}^{d}X_{ij} + \tau Z_i + \epsilon_i$, where $\tau = 0.5$ and $\epsilon_i \sim N(0,1)$. We set $n_1 = n_0 = n/2$ and $n = 500$. For the rerandomization, we consider the acceptance probability $p_a = 0.01$, so the threshold $a$ is the $p_a$-th quantile of a $\chi^2_d$ distribution.

We perform Fisher randomization tests with a significance level $\alpha = 0.05$ and construct confidence intervals with a nominal coverage rate of $95\%$. To estimate computational runtime and the bias and standard deviation of the treatment effect estimator, we conduct $B=1000$ independent simulations. Within each simulation, we utilize the specific rerandomization algorithm to generate the balanced treatment assignment. To conduct randomization-based inference and construct confidence intervals, we generate a reference set of $B_{\text{frt}}=100$ additional balanced randomizations for each simulation to approximate the null distribution. Since PSRR and ARR scale with the dimension size, we do not report any results that depend on randomization-based inference for these methods.

In Figure~\ref{fig:competitor_lowdimension}, we consider when $d$ is small. In the left panel, each line considers the average time in seconds to find a balanced randomization by each rerandomization method as a function of the dimension $d$, or the time to draw a randomization completely at random. The shade around each line is a bootstrap 95\% confidence interval calculated from the $B=1000$ randomizations. The right panel represents the bias of the treatment effect estimator as a function of the dimension size. The lines are the average bias. The shaded region represents the standard deviation of the estimated bias. We find that initially, ARR is the slowest to find a balanced randomization, followed by our proposed method. However, as the dimension increases, PSRR turns into the slowest rerandomization method, and our proposed method is the fastest between the rerandomizations. All of the methods have similar biases and standard deviations of the estimated treatment effect, which are lower that the standard deviation from CR.

\begin{figure}[h]
\centering
\includegraphics[width=0.85\textwidth]{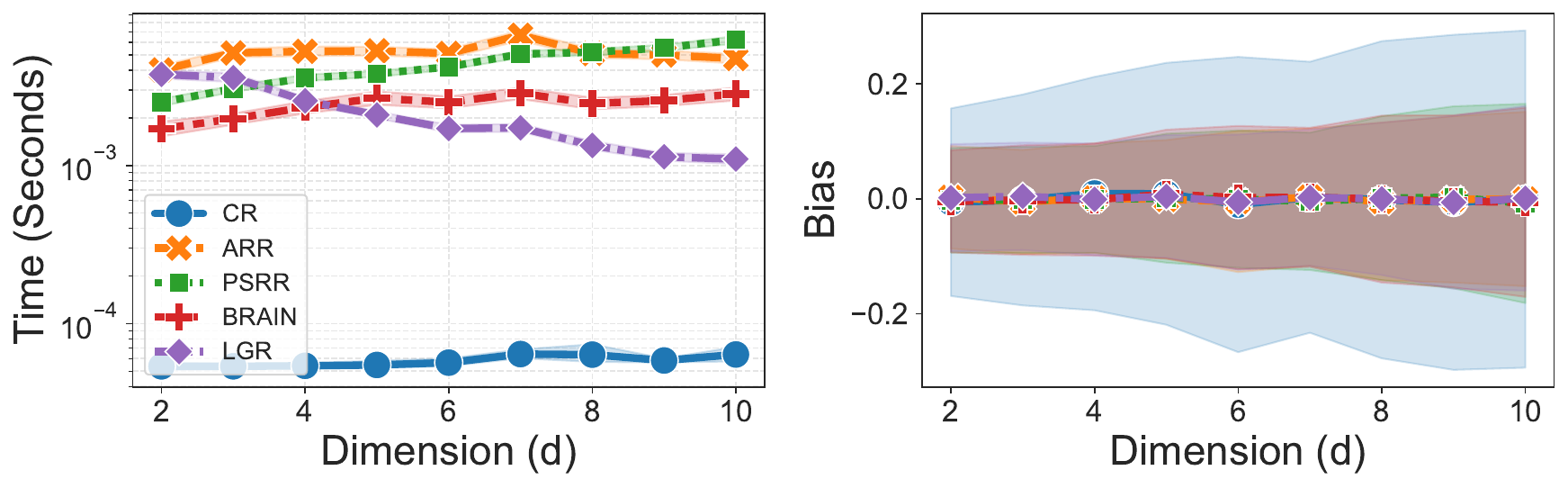}
\caption{The left panel represents the time to generate a balanced randomization using each rerandomization method or a randomization completely at random (blue line). The shaded region around each line is a bootstrap 95\% confidence interval. The right panel shows the bias (lines) and standard deviation (shaded region) of the difference-in-means treatment effect estimator with each rerandomization method. We find that for the smallest dimensions ARR is the slowest method to find a balanced randomization, followed by our proposed method. However, as the dimension size increases, PSRR turns into the slowest method and LGR is the fastest method.}
\label{fig:competitor_lowdimension}
\end{figure}

We extend this simulation to higher dimension size, and present it in Figure~\ref{fig:competitor_highdimension}.  We find that our proposed method is the fastest to find a balanced randomization while PSRR is the slowest. Interestingly, our method presents a U-shape curve in the left plot. This might happen because calculating the gradient of the soft relaxation is an overhead in low dimensions making the algorithm slower, but proves to be computationally more efficient in higher dimensions. On the other hand, PSRR works as a ``random walk" on the treatment space with a step of size one (only one unit on each treatment arm is swapped at each iteration). Hence, it takes longer to find the balanced region in the treatment space.

\begin{figure}[h]
\centering
\includegraphics[width=0.85\textwidth]{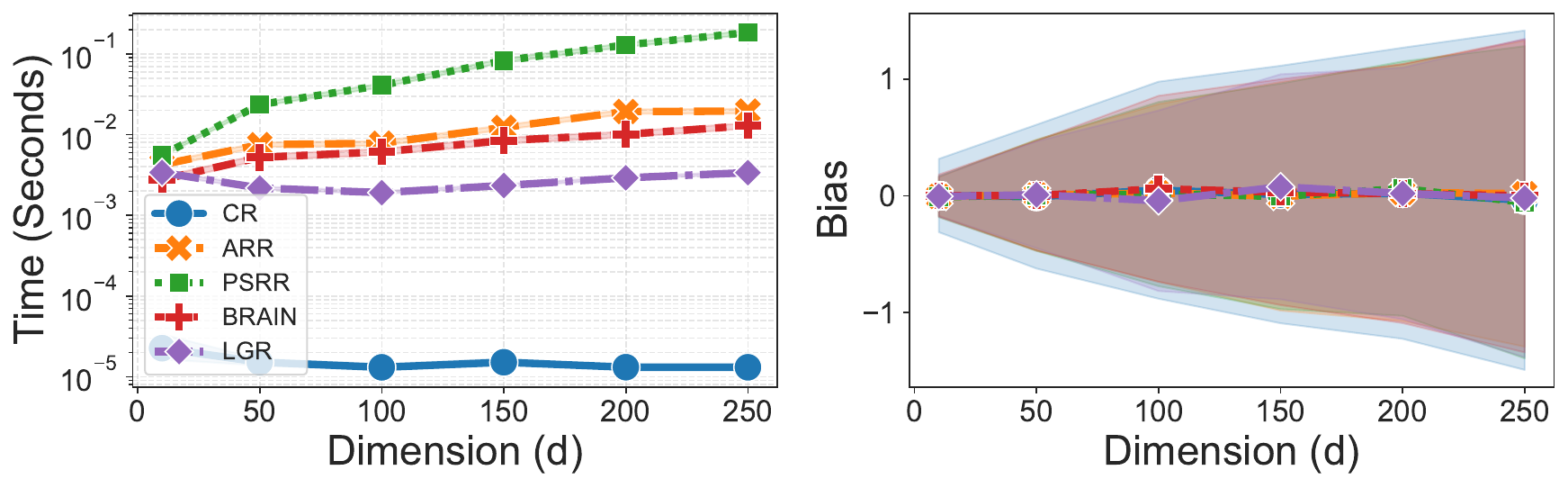}
\caption{The left panel represents the time to generate a balanced randomization using each rerandomization method or a randomization completely at random (blue line). The shaded region around each line is a bootstrap 95\% confidence interval. The right panel shows the bias (lines) and standard deviation (shaded region) of the difference-in-means treatment effect estimator with each rerandomization method. For the higher dimensions, LGR is the fastest to generate a balanced randomization. Interestingly, it has a U-shape, suggesting that for low dimensions the overhead of calculating gradients, while for it is computationally beneficial to calculate the gradients.}
\label{fig:competitor_highdimension}
\end{figure}

Next, we do randomization-based inference with LGR, BRAIN, and CR and present the results in Figure~\ref{fig:inference}. In the left panel, we show the coverage probability of each method as a function of the dimension size $d$. Nominal coverage at 95\% is denoted by the dashed horizontal black line. While on the right panel, we show the power of each method as a function of the dimension size. Notice that all method achieve nominal coverage, while BRAIN and LGR are more powerful than CR, following the rerandomization literature.

\begin{figure}[h]
\centering
\includegraphics[width=0.80\textwidth]{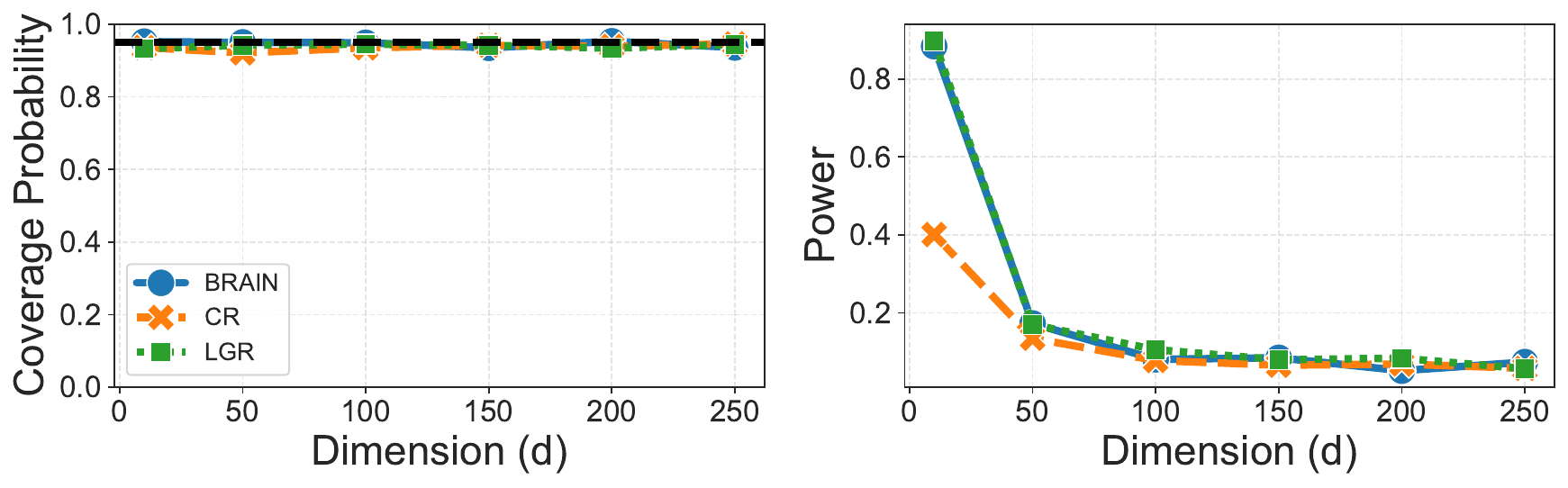}
\caption{The left panel represents the coverage probability of each method. Nominal coverage is represented by the dashed horizontal black line at 95\%. Meanwhile the right panel shows the power. All methods achieve nominal coverage, while the rerandomization methods are more powerful than complete randomization.}
\label{fig:inference}
\end{figure}

In Appendix \ref{app:simulations}, we conduct a sensitivity analysis of our proposed rerandomization method with respect to its temperature $\delta$ and learning rate $\eta$ parameters. We find that extreme values of the parameters can be detrimental to the rerandomization method and make it slower than the established rerandomization algorithms in the literature. 

\section{Conclusion}\label{sec:conclusion}
Rerandomization is a powerful tool for improving precision in randomized experiments by achieving covariate balance at the design stage \cite{morgan2012, li2018}. However, the standard implementation via acceptance–rejection sampling suffers from an exponential computational bottleneck as the dimension of covariates increases \cite{ribeiro2025does}. Although recent alternatives such as PSRR and BRAIN have been proposed to mitigate this curse of dimensionality, these methods rely on discrete procedures—iterative local swaps or constrained optimization—that lack direct guidance from gradient information, limiting their efficiency in high-dimensional spaces.

In this paper, we propose Langevin-Gradient Rerandomization, a novel rerandomization method that exploits a continuous relaxation of the treatment assignment space to navigate toward balanced randomizations via Stochastic Gradient Langevin Dynamics. By reformulating the discrete assignment problem into a continuous optimization landscape, we leverage gradient information about covariate imbalance to guide the search process efficiently. We proved that despite sampling non-uniformly from the set of balanced randomizations, the difference-in-means treatment effect estimator remains unbiased (Theorem~\ref{thm:unbiased}) and achieves variance reduction comparable to standard rerandomization methods (Theorem~\ref{thm:variance}). To ensure valid finite-sample inference under non-uniform sampling, we employed Fisher Randomization Tests with confidence interval inversion, providing exact hypothesis tests conditional on the LGR sampling mechanism. Extensive empirical simulations across dimensions demonstrate that LGR achieves orders of magnitude speedups over existing methods in high-dimensional settings while maintaining unbiased and precise estimation as previous rerandomization methodologies. LGR leads to nominal coverage inference as powerful as previous rerandomization methods.

Future research directions include extending LGR to more general differentiable covariate balance metrics beyond the Mahalanobis distance, such as quadratic forms \citep{schindl2024}. Additionally, LGR could be generalized to other experimental designs where rerandomization is required, such as sequential designs where units arrive over time \cite{zhou2018} and cluster randomized trials where balance must be maintained within and between clusters \citep{lu2023}. Further refinements could include adaptive learning strategies to adjust LGR's parameters based on gradient magnitudes and covariance structure, eliminating the need for manual hyperparameter tuning.

\clearpage

\bibliographystyle{plainnat}
\bibliography{main.bib}

\appendix


\section{Proofs}\label{app:proofs}
\subsection{Proof of Theorem~\ref{thm:unbiased}}\label{proof:unbiased}
\begin{proof}
Let $\Omega$ denote the set of all possible trajectories (chains) of the latent variable vectors generated by the LGR algorithm: $\left(\theta^{(0)},...,\theta^{(t)},...,\theta^{(T)}\right),$ where $\theta^{(t)}\in\mathbb{R}^{n}$. Let $\pi$ be any permutation of the indices $\{1,...,n\}$. Consider any realized chain $\vartheta=\left(\vartheta^{(0)},...,\vartheta^{(T)}\right)\in\Omega.$ We compare the probability of this chain to its negation counterpart $-\vartheta=\left(-\vartheta^{(0)},...,-\vartheta^{(T)}\right)$.

The algorithm initializes $\theta^{(0)}\sim N(0,I_{n})$. The standard multivariate normal distribution is exchangeable, $\Prob\left(\theta^{(0)}=\vartheta^{(0)}\right)=\Prob\left(\theta^{(0)}=-\vartheta^{(0)}\right)$.

The transition from $\theta^{(t-1)}$ to $\theta^{(t)}$ is governed by the SGD update:
\begin{equation*}
    \theta^{(t)}=\theta^{(t-1)}-\nabla J\left(\theta^{(t-1)}\right)+\xi_{t} \quad \xi_{t}\sim N(0,I_{n}).
\end{equation*}

The objective function $M(\theta)$ (soft Mahalanobis distance) is a Mahalanobis distance constructed from a sigmoid function $\sigma$. Hence, $M(\theta)$ is an even function, and its gradient is an odd function, meaning $\nabla M(-\theta) = - \nabla M(\theta)$.

As in the initialization of he algorithm, since the noise $\xi_{t}$ is also isotropic (exchangeable), the transition probability density satisfies: 
\begin{equation*}
    \Prob\left(\theta^{(t)}=\vartheta^{(t)}|\theta^{(t-1)}=\vartheta^{(t-1)}\right)=\Prob\left(\theta^{(t)}=-\vartheta^{(t)}|\theta^{(t-1)}=-\vartheta^{(t)}\right)
\end{equation*}

Combining all of these
\begin{align*}
\Prob(\text{Chain}=\vartheta) &= \Prob\left(\theta^{(0)}\right)\prod_{t=1}^{T}\Prob\left(\theta^{(t)}|\theta^{(t-1)}\right) \\
&= \Prob\left(-\theta^{(0)}\right)\prod_{t=1}^{T}\Prob\left(-\theta^{(t)})|-\theta^{(t-1)}\right) \\
&= \Prob(\text{Chain}=-\vartheta)
\end{align*}

This implies that for every trajectory the algorithm takes, the negative trajectory is equally likely. The final assignment $Z$ is a deterministic function of the final state $\theta^{(T)}$ (specifically, $Z_{i}=1$ if $\theta_{i}^{(T)}$ is among the top $n_{1}$ values). If $\theta_{i}^{(T)}$ is in the top $n_{1}$, then the $i$-th element of $-\theta^{(T)}$ is in the bottom $n_{1}$. Thus, $Z\left(-\theta^{(T)}\right)=1-Z\left(\theta^{(T)}\right)$ and the marginal distribution of the final assignment $Z$ must be
$$\Prob(Z=z)=1-\Prob(Z=z)$$

Summing over all possible permutations $\pi,$ the marginal probability that any specific unit $i$ is assigned to treatment must be identical for all units:
$$\Prob(Z_{i}=1)=\Prob(Z_{i}=0)=\frac{1}{2}, \forall i = 1,\dots,n$$

Finally, with $\E(Z_i) = \Prob(Z_{i}=1)=
1/2$, the expectation of the difference-in-means estimator is:
\begin{align*}
\E[\hat{\tau}] &= \E\left[\frac{2}{n}\sum_{i=1}^{n} Z_{i}Y_{i}(1)-\frac{2}{n}\sum_{i=1}^{n}(1-Z_{i})Y_{i}(0)\right] \\
&= \frac{2}{n}\sum_{i=1}^{n} Y_{i}(1) \E(Z_{i}) -\frac{2}{n}\sum_{i=1}^{n}Y_{i}(0)\E(1-Z_{i}) \\
&= \frac{1}{n}\sum_{i=1}^{n}(Y_{i}(1)-Y_{i}(0))=\tau.
\end{align*}

\end{proof}

\subsection{Proof of Theorem~\ref{thm:variance}}\label{proof:variance}
\begin{proof}
We can define
\begin{equation*}
    W = \sqrt{\frac{n}{n_1n_0}}S^{-1}_X\left(\mean{X}_1 - \mean{X}_0\right) = \sqrt{\frac{n}{n_1n_0}}S^{-1}_X X^\prime\left(Z - \frac{n_1}{n}\1_n\right)
\end{equation*}
and $M = W'W = \sum_{j=1}^{d}W_j^2$. Therefore, under Assumption~\ref{ass:same_size}, $\E(W) = 0$ and $\E(W_j) = 0, \forall j = 1,\dots,d$.

Note that if we exchange any $W_i$ and $W_j$, $M$ does not change. As a result, we have
\begin{equation*}
    \var(W_j) = \E(W_j^2) = \frac{1}{d}\E\left[\sum_{j=1}^{d}W_j^2\right] = \frac{1}{d}\E(M) \leq \frac{a}{d}.
\end{equation*}

Moreover, if we change the sign of any $W_j$, the Mahalanobis distance $M$ does not change. Hence this does not affect the distribution of W, and the distribution of $W_j$. Thus, $W_j|W_i$ and $-W_j|W_i$ are identically distributed. This implies,
\begin{equation*}
    \cov(W_j, W_i) = \E(W_jW_i) = \E\left[\E\left[W_jW_i|W_i\right]\right]= \E\left[W_i\E\left[W_j|W_i\right]\right]= \E\left[W_i \times 0\right] = 0
\end{equation*}

Therefore, $\cov(W) = \var(W_j)I_d$,
\begin{align*}
    \cov\left(\mean{X}_1 - \mean{X}_0\right) &= \cov\left(\sqrt{\frac{n_1n_0}{n}}S_XW\right)\\
    &= \frac{n_1n_0}{n}S_X\cov(W)S_X \\
    &= \frac{n_1n_0}{n}S_X\var(W_j)S_X \\
    & \leq \frac{a}{d} \frac{n_1n_0}{n}S^2_X \\
    &= \frac{a}{d}\varcr(\tauest{})
\end{align*}
and $\beta^\prime\cov\left(\mean{X}_1 - \mean{X}_0\right)\beta \leq \frac{a}{d}R^2\varcr(\tauest{})$.

Finally, by Assumption~\ref{ass:linearity} the difference-in-means estimator can be expressed as
\begin{equation*}
    \tauest{} = \tau + \beta^\prime\left(\mean{X}_1 - \mean{X}_0\right) + \left(\mean{\epsilon}_1 - \mean{\epsilon}_0\right)
\end{equation*}
where $\epsilon_1 = \sum_{i : Z_i = 1} \epsilon_i/n_1$ and $\epsilon_0 = \sum_{i : Z_i = 0} \epsilon_i/n_0$. Since $\beta_0 + \beta^\prime X_i$ is the projection of $Y_i(0)$ onto $(1,X)$, $\mean{X}_1 - \mean{X}_0$ and $\mean{\epsilon}_1 - \mean{\epsilon}_0$ are uncorrelated. By Assumption~\ref{ass:normality}, $\tauest{}$ and $\mean{X}_1 - \mean{X}_0$ are normally distributed, so $\mean{X}_1 - \mean{X}_0$ and $\mean{\epsilon}_1 - \mean{\epsilon}_0$ are independent. Since LGR does not affect $\mean{\epsilon}_1 - \mean{\epsilon}_0$, we have
\begin{align*}
    \varlgr(\tauest{}) &= \beta^\prime\cov\left(\mean{X}_1 - \mean{X}_0\right)\beta + \var\left(\mean{\epsilon}_1 - \mean{\epsilon}_0\right) \\
    &\leq \frac{a}{d}R^2\varcr(\tauest{}) + (1-R^2)\varcr(\tauest{}) \\
    &= \left[1-\left(1-\frac{a}{d}\right)R^2\right]\varcr(\tauest{})
\end{align*}

Hence,
\begin{equation*}
    \frac{\varcr(\tauest{}) - \varlgr(\tauest{})}{\varcr(\tauest{})} \geq \left(1 - \frac{a}{d}\right)R^2.
\end{equation*}
\end{proof}

\section{Additional Simulations}\label{app:simulations}

In addition to the simulation results in Section~\ref{sec:simulations}, we conduct a sensitivity analysis to understand whether the performance of LGR is affected by the choice of its parameters.

Both panels in the left side of Figure~\ref{fig:sensitivity} show the results for LGR when varying its temperature $\delta \in \{0.01, 0.1, 0.5, 1.0, 10.0\}$, while considering the default value for the learning rate $\eta = 1$, and compare to the performance of BRAIN with its default parameters. The lines on the top left panel show the time in seconds for each setting to sample a balanced randomization, with its shaded region corresponding to a bootstrap 95\% confidence interval. Notice that for extreme values of $\delta$, LGR's performance is undermined, because the soft assignments defined in Equation (\ref{eq:soft_treatment}) take extreme values, making the gradient in Equation (\ref{eq:grad-M}) unstable and uninformative. Meanwhile, the bottom left panel shows the bias (lines) and standard deviation (shaded region) of the difference-in-means treatment effect estimator for each setting of the rerandomization methods. Note that all of them are similar, and hence the temperature does not seem to affect it.

\begin{figure}[H]
\centering
\includegraphics[width=0.9\textwidth]{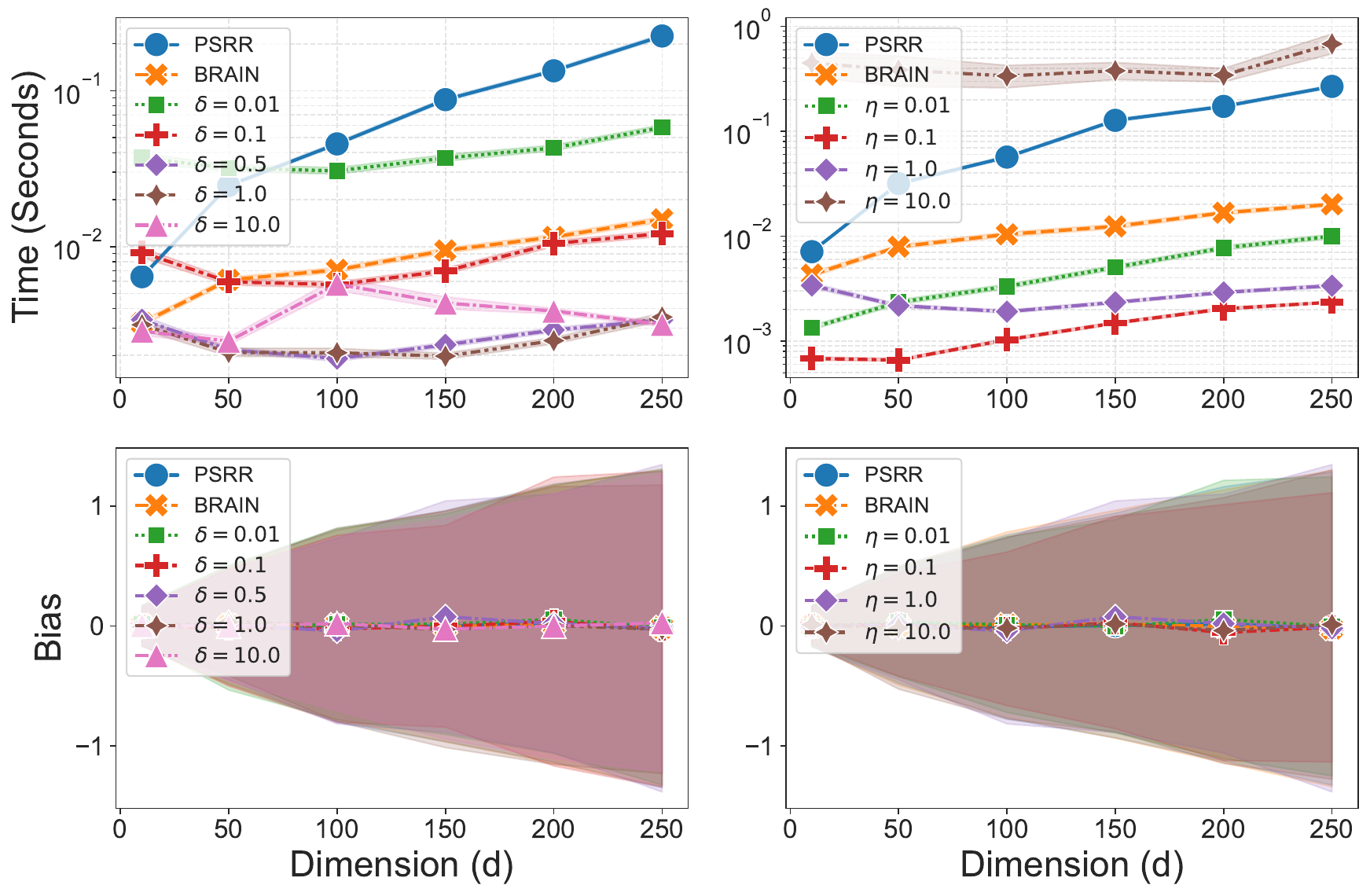}
\caption{Sensitivity analysis plot of the LGR method with respect to its temperature $\delta$ and $\eta$. Both panels in the left side show the results for LGR when varying its temperature $\delta \in \{0.01, 0.1, 0.5, 1.0, 10.0\}$, while considering the default value for the learning rate $\eta = 1$, and BRAIN method with its default parameters. The panels on the right side consider a fixed temperature at its default value $\delta = 0.5$ , but varying $\eta \in \{0.01, 0.1, 1.0, 10.0\}$.The lines on the top panels show the time in seconds for each setting to sample a balanced randomization, with its shaded region corresponding to a bootstrap 95\% confidence interval. Notice that for extreme values of $\delta$ and large values of $\eta$, LGR's performance is undermined. Meanwhile, the bottom right panel shows the bias (lines) and standard deviation (shaded region) of the difference-in-means treatment effect estimator for each setting of the rerandomization methods. Note that all of them are similar, and hence the parameters values do not seem to affect it.}
\label{fig:sensitivity}
\end{figure}

The other two panels on the right side of Figure~\ref{fig:sensitivity} show similar results. In this case, they consider LGR with default temperature $\delta = 0.5$ but varying learning rate $\eta \in \{0.01, 0.1, 1.0, 10.0\}$.  The lines in the top right panel show the time in seconds for each setting to sample a balanced randomization, with its shaded region corresponding to a bootstrap 95\% confidence interval. Notice that for large values of $\eta$, LGR's performance is undermined, this is because the SGLD update step is taking large step sizes, which ultimately does not exploit the gradient information. Meanwhile, the bottom right panel shows the bias (lines) and standard deviation (shaded region) of the difference-in-means treatment effect estimator for each setting of the rerandomization methods. Note that all of them are similar, and hence the learning rate does not seem to affect it.

\end{document}